\documentclass[12pt]{iopart}
\usepackage{setstack}
\usepackage{amsthm}
\usepackage{amsfonts}
\usepackage{amssymb}
\usepackage{graphicx}
\usepackage{subfig}
\usepackage{bm}
\newcommand{\ud}{\,\textrm{d}}
\newcommand{\RR}{\mathbb{R}}

\newcommand{\bxi}{\boldsymbol{\xi}}
\newcommand{\bx}{\boldsymbol{x}}
\newcommand{\by}{\boldsymbol{y}}
\newcommand{\balpha}{\boldsymbol{\alpha}}
\newcommand{\bz}{\boldsymbol{z}}
\newcommand{\bv}{\boldsymbol{v}}
\newcommand{\bu}{\boldsymbol{u}}
\newcommand{\be}{\boldsymbol{e}}

\newcommand{\bzero}{\boldsymbol{0}}
\newcommand{\bzeta}{\boldsymbol{\zeta}}
\newcommand{\bnabla}{\boldsymbol{\nabla}}

\newcommand{\HH}[1]{\mathcal{H}_{\text{#1}}}
\newcommand{\spn}{\text{Span}}
\newtheorem{theorem}{Theorem}
\newtheorem{lemma}[theorem]{Lemma}

\newtheorem{corollary}[theorem]{Corollary}

\graphicspath{{figures/}}

\begin{document}

\title[Calogero-Moser Systems as a Diffusion-Scaling Transform of Dunkl Processes]{Calogero-Moser Systems as a Diffusion-Scaling Transform of Dunkl Processes on the Line}

\author{Sergio Andraus$^1$, Makoto Katori$^2$ and Seiji Miyashita$^1$}
\address{$^1$ Department of Physics, Graduate School of Science, University of Tokyo,
7-3-1 Hongo, Bunkyo-ku, Tokyo 113-0033}
\address{$^2$ Department of Physics, Graduate School of Science and Engineering, Chuo University, 1-13-27 Kasuga, Bunkyo-ku, Tokyo 112-8551}
\ead{andraus@spin.phys.s.u-tokyo.ac.jp}

\begin{abstract}

The Calogero-Moser systems are a series of interacting particle systems on one dimension that are both classically and quantum-mechanically integrable. Their integrability has been established through the use of Dunkl operators (a series of differential-difference operators that depend on the choice of an abstract set of vectors, or root system). At the same time, Dunkl operators are used to define a family of stochastic processes called Dunkl processes. We showed in a previous paper that when the coupling constant of interaction of the symmetric Dunkl process on the root system $A_{N-1}$ goes to infinity (the freezing regime), its final configuration is proportional to the roots of the Hermite polynomials. It is also known that the positions of the particles of the Calogero-Moser system with particle exchange become fixed at the roots of the Hermite polynomials in the freezing regime. Although both systems present a freezing behaviour that depends on the roots of the Hermite polynomials, the reason for this similarity has been an open problem until now. In the present work, we introduce a new type of similarity transformation called the diffusion-scaling transformation, in which a new space variable is given by a diffusion-scaling variable constructed using the original space and time variables. We prove that the abstract Calogero-Moser system on an arbitrary root system is a diffusion-scaling transform of the Dunkl process on the same root system. With this, we prove that the similar freezing behaviour of the two systems on $A_{N-1}$ stems from their similar mathematical structure.

\end{abstract}

\pacs{05.40.-a, 02.30.Ik, 02.50.Ga}

\submitto{\JPA}

\maketitle

\section{Introduction}\label{intro}

The Calogero-Moser (CM) systems are a family of exactly-solvable, interacting multiple particle one-dimensional systems. They have been of great interest in both physics and mathematics because they are one of the few interacting particle system families that can be solved exactly. The first such system was formulated by Calogero \cite{calogero71}, and it is composed of $N$ particles on a line confined by a harmonic potential, in which the particles repel each other through a potential proportional to the inverse of the squared distance between them. A similar system on the unit circle instead of the line was considered by Sutherland \cite{sutherland71A}. CM systems defined on the unit circle are commonly known as Calogero-Moser-Sutherland (CMS) systems. Later, Moser \cite{moser75} established the integrability of the CM and CMS systems. These systems are categorized by the root system associated to them. A root system is a finite set of vectors that is closed under reflections; in other words, when one of its elements is reflected by another element, the result is an element of the root system as well (see Section~\ref{Preliminaries}). $A_{N-1}$ is one such set of vectors. Olshanetsky and Perelomov proved that the systems defined in \cite{calogero71} and \cite{sutherland71A} are intimately related to the root system $A_{N-1}$, and that their analogues on all classical root systems are integrable \cite{olshanetskyperelomov76,olshanetskyperelomov83}. Polychronakos considered the inclusion of a particle exchange operator on the CM system first considered by Calogero \cite{polychronakos92}; this operator is also related to $A_{N-1}$, so the CM system with particle exchange interaction is an extension of the CM system defined on \cite{calogero71}. He also studied its behavior when the potential strength parameter tends to infinity. This operation is known as the ``freezing trick'' \cite{polychronakos93}. When the freezing trick is applied, the positions of the particles of the CM system with particle exchange interaction are fixed to the roots of the Hermite polynomials, so only the exchange interaction between particles remains as the non-constant part of the original system Hamiltonian. Because the particle exchange operator can be regarded as a spin interaction operator \cite{hikamiwadati93}, the resulting system is called the Polychronakos-Frahm spin chain. Its integrability was consequently proved by Frahm \cite{frahm93}. The spectrum of the CM system with exchange operators on other root systems has been calculated by Khastgir, Pocklington and Sasaki \cite{khastgir00}. 

The integrability of the CM systems has been established by various methods, one of which makes use of Dunkl operators \cite{bernard93,forrester10}. These are differential difference operators that depend on the choice of a root system and a set of parameters called ``multiplicites'', and they were defined for the study of multivariate orthogonal polynomials \cite{dunkl89,dunklxu}. When a similarity transformation is applied on the CM system Hamiltonian using its ground-state eigenfunction (see Section~\ref{similaritycm}), the result is an operator that has a very simple form when it is expressed in terms of the Dunkl operators \cite{bakerdunklforrester}. At the same time, Dunkl operators are used to define Dunkl processes as a generalization of Brownian motion in the following sense: the Kolmogorov backward equation (KBE) that generates Brownian motion is the heat equation; if all the spatial derivatives in it are replaced by Dunkl operators, one obtains the Dunkl heat equation. Then, Dunkl processes are defined as the stochastic processes generated by the Dunkl heat equation \cite{rosler98,roslervoit98}. Dunkl processes are discontinuous stochastic processes with a drift term in general, but their ``radial part", or the part that is invariant under the reflection operators of the root system, is continuous.

Demni \cite{demni08A} noted that the radial Dunkl process of type $A_{N-1}$ is equivalent to Dyson's Brownian motion \cite{dyson62B}. We used this equivalence in our previous work \cite{andrauskatorimiyashita12} to give a physical interpretation to the Dunkl process of type $A_{N-1}$. It represents a system of $N$ Brownian particles in one dimension that repel each other through a logarithmic potential and exchange positions spontaneously. Additionally, we found that the coupling constant for this interaction can be understood as the inverse temperature. More importantly, we used the mathematical tools of Dunkl processes to investigate the freezing properties of the radial Dunkl processes of type $A_{N-1}$, and the equivalent Dyson's Brownian motion. We found that the (suitably scaled) trajectories of the Brownian particles are given by the roots of the $N$-th Hermite polynomial multiplied by the square root of the process time when the temperature tends to zero. This is similar to the behaviour of the CM system with particle exchange interaction when the freezing trick is applied. The close relationship that exists between Dunkl operators and the CM systems indicates that the similarity in the freezing behaviour of the CM system and the radial Dunkl process of type $A_{N-1}$ is not accidental. The purpose of this work is to find a clear relationship between the CM systems and Dunkl processes that explains their similarity in the freezing regime.

Our main result (Theorem~\ref{correspondencer}) is that one can use the diffusion-scaling transformation on the Kolmogorov forward equation (KFE) of any Dunkl process to obtain the Schr{\"o}dinger equation of a CM system evolving on imaginary time, both on the same root system. The diffusion-scaling transformation is a similarity transformation in which a new space variable is defined by dividing the original space variable by the square root of the process time. Due to its importance in the derivation of our main result, we explain the diffusion-scaling transformation in more concrete terms as follows. Suppose that the frequency of the harmonic potential of the CM system is $\omega$. Then, we apply the variable substitution given by
\begin{equation}
(t,\bx)\to(\tau,\bzeta)=\Bigg(\frac{\ln t}{2\omega},\frac{\bx}{\sqrt{2\omega t}}\Bigg)\label{substitutionr}
\end{equation}
to the KFE of the Dunkl process. After that, we use the function $\rme^{-W(\tau,\bxi)}$, with
\begin{equation}
W(\tau,\bzeta)=\frac{1}{2}\omega\sum_{i=1}^N\zeta_i^2-\frac{1}{2}\ln w_k(\bzeta)+\omega N\tau\label{Wr},
\end{equation}
to perform a similarity transformation on the result. Here, $w_k(\bx)$ is a weight function that depends on the root system being considered, and its definition is given by \eref{weightk}. Then, the result is that by applying the diffusion-scaling transformation to a Dunkl process (with space and time variables $\bx$ and $t$, respectively) we obtain a CM system on the same root system (with variables $\bzeta$ and $\tau$).

In Section~\ref{Preliminaries}, we review the definition of the reflection operator, the concept of root system and the definition of the multiplicity function. We also recall the Hamiltonian of the CM system with exchange operators on an arbitrary root system $R$ and the particular case of the CM system on $A_{N-1}$, as well as the definition of the Dunkl processes and their KBEs and KFEs. In Section~\ref{similaritycm}, we review how a similarity transformation of the CM Hamiltonian using its ground-state eigenfunction yields a Dunkl Laplacian (a Laplacian operator with Dunkl operators instead of partial derivatives) plus a confinement term. In Section~\ref{ideakt}, we review the main idea that gives the correspondence between Brownian motion and the quantum harmonic oscillator in one dimension. This idea, applied to the Dunkl process on a root system $R$, is the basis for the diffusion-scaling transformation we use to derive our main result, Theorem~\ref{correspondencer} in Section~\ref{resultr}. That is, we prove that by using the diffusion-scaling transformation defined by the relations given in \eref{substitutionr} and \eref{Wr}, the KFE of a Dunkl process is transformed into the Schr{\"o}dinger equation of the CM system. We explain the similar freezing behaviour of the Dunkl process and the CM system on the root system $A_{N-1}$ by applying Theorem~\ref{correspondencer} in Section~\ref{resulta}. We conclude by discussing these results and by outlining possible topics for future study in Section~\ref{conclusions}.

\section{Preliminaries}\label{Preliminaries}

\subsection{Root Systems and the Multiplicity Function}

We begin by defining the reflection operator by the equation
\begin{equation}\label{reflection}
\sigma_{\balpha}\bx=\bx-2\frac{\bx\cdot \balpha}{\alpha^2}\balpha.
\end{equation}
This operator acts on the vector $\bx$ by reflecting it through the hyperplane defined by the vector $\balpha$ in $N$ dimensions ($\balpha,\bx\in\RR^N$). Note that we write vectors with bold symbols, while we denote their squared norm by $|\balpha|^2=\balpha\cdot\balpha=\alpha^2$.

A root system is defined as a finite set of vectors such that, when its elements, called roots, are reflected by any root, the resulting vector also belongs to the set. In other words, a root system $R$ is defined by the property that $\sigma_{\balpha}\bxi\in R$ for any $\balpha,\bxi\in R$. In addition, a root system is called reduced if, for all $\balpha\in R$, the statement $b\balpha\in R$ with $b\in\RR$ implies that $b=\pm 1$. We will assume that all the root systems considered here are reduced. We also assume that no root is the zero vector, because a reflection along the zero vector is undetermined. 

For every root system, there exists a set of base vectors with which every root is given by a linear combination of the base vectors with all the coefficients being either positive or negative. The roots in this base are called simple roots. This base is not unique, but it always divides the root system into two subsets: the positive subsystem $R_+$ (roots generated by positive coefficients) and the negative subsystem $R_-$ (negative coefficients), and both subsystems have the same number of elements. 

An example of this is given by the root system $B_2=\{\pm \be_1, \pm \be_2, \pm(\be_1-\be_2), \pm(\be_1+\be_2)\}$ in $\RR^2$, where $\be_i$ denotes the $i$-th canonical base vector, $i=1,2$. If we choose $\be_1$ and $\be_2-\be_1$ as the simple roots, the positive subsystem becomes $\{\be_1,\be_2,\be_2+\be_1,\be_2-\be_1\}$. The negative subsystem is formed by the same roots multiplied by $-1$. In general, $R_-=\{-\balpha:\balpha\in R_+\}$ (this is sometimes written as $R_-=-R_+$.)

The multiplicity function, $k(\balpha)$, $\balpha\in R$, is a series of parameters, called multiplicities, assigned to each disjoint part of a root system in the following sense: if there exist roots $\balpha, \bxi, \bzeta$ such that $\sigma_{\balpha}\bxi=\bzeta$, then $k(\bxi)=k(\bzeta)$. Consequently, the multiplicities assigned to two different roots is different only if they cannot be related by a series of reflections. 

In the case $R=B_2$, there is no root that reflects the roots $\pm \be_1$ and $\pm \be_2$ into the roots $\pm(\be_1-\be_2)$ and $\pm(\be_1+\be_2)$. Therefore, the multiplicity function for $B_2$ can take only two different values: $k(\pm \be_1)=k(\pm \be_2)=k_1$ and $k(\pm(\be_1-\be_2))=k(\pm(\be_1+\be_2))=k_2$. That is, the root system $B_2$ has two multiplicities. In general, we will assume that all multiplicities are non-negative real numbers.

It is useful to define the sum of multiplicities over the positive subsystem as
\begin{equation}
\gamma=\sum_{\balpha\in R_+}k(\balpha)=\frac{1}{2}\sum_{\balpha\in R}k(\balpha).
\end{equation}
Note that this value does not depend on the choice of simple subsystem. For example, for $B_2$, $\gamma=2(k_1+k_2)$ for all possible positive subsystems.

We will focus on the root system $A_{N-1}$ in some of the later sections. This root system is given by $A_{N-1}=\{(\be_i-\be_j)/\sqrt{2}:1\leq i\neq j\leq N\}$. Choosing the simple roots
\begin{equation}
(\be_{i+1}-\be_i)/\sqrt{2},\ i=1,\ldots,N-1\label{simpleroots}
\end{equation}
produces the positive subsystem
\begin{equation}
A_{N-1,+}=\{\balpha_{ij}=(\be_i-\be_j)/\sqrt{2}:1\leq j<i\leq N\}.\label{positivesubsystem}
\end{equation}
Additionally, this root system has only one multiplicity, $k$, and we have
\begin{equation}
\gamma=\frac{N}{2}(N-1)k.
\end{equation}
Most importantly, the effect of the reflections along the roots of $A_{N-1}$ is that of exchanging the reflected vector's components. Using the notation $\sigma_{ij}=\sigma_{\balpha_{ij}}$, the expression $\sigma_{ij}\bx$ denotes the vector obtained by exchanging the $i$-th and $j$-th components of the vector $\bx$. For details on the above, see Appendix~A on \cite{andrauskatorimiyashita12}.

Finally, we define the following weight function:
\begin{equation}
w_k(\bx)=\prod_{\balpha\in R}|\balpha\cdot\bx|^{k(\balpha)}.\label{weightk}
\end{equation}
This function is invariant under reflections along the roots of $R$, and it is harmonic, \emph{i.e.}, 
\begin{equation}
\Delta^{(x)}w_k(\bx)=\sum_{i=1}^{N}\frac{\partial^2}{\partial x_i^2}w_k(\bx)=0.
\end{equation}
This function, $w_k(\bx)$, will be used in the transformation considered in Theorem~\ref{correspondencer} in Section~\ref{resultr}.

\subsection{Calogero-Moser systems}\label{PreliminariesCalogeroMoserSystems}

Under an arbitrary root system $R$, the CM systems on a line with a harmonic background potential and an inverse-square repulsion potential are given by the Hamiltonian (see, \emph{e.g.}, \cite{khastgir00})
\begin{equation}
\mathcal{H}_{\text{CM}}^R=-\frac{1}{2}\Delta^{(x)}+\sum_{\balpha\in R_+}\frac{\alpha^2}{2}\frac{k(\balpha)[k(\balpha)-\sigma_\alpha]}{(\balpha\cdot\bx)^2}+\frac{\omega^2}{2}\sum_{i=1}^Nx_i^2,\label{generalcm}
\end{equation}
where all particles have been chosen to be of unit mass, and we have taken $\hbar =1$. 
%This system's ground-state energy is given by $E_{\text{GGS}}=\omega[\gamma+N/2]$. 
In particular, when the root system is $A_{N-1}$ and the harmonic frequency $\omega$ is set equal to the multiplicity $k$, the CM system with  particle exchange interaction considered in \cite{polychronakos93} is recovered:
\begin{equation}
\mathcal{H}_{\text{CM}}^{A_{N-1}}=-\frac{1}{2}\Delta^{(x)}+\sum_{1\leq i<j\leq N}\frac{k(k-\sigma_{ij})}{(x_i-x_j)^2}+\frac{k^2}{2}\sum_{i=1}^{N}x_i^2\label{xcm}.
\end{equation}

\subsection{Dunkl processes}

We define the gradient operator as $\bnabla^{(x)}=(\frac{\partial}{\partial x_1},\ldots,\frac{\partial}{\partial x_N})$. Under an arbitrary root system $R$, the Dunkl operator in the direction $\bxi$ is given by \cite{dunkl89}
\begin{equation}
T_{\bxi} f(\bx)=\bxi\cdot\bnabla^{(x)} f(\bx)+\sum_{\balpha \in R_+}k(\balpha)\frac{f(\bx)-f(\sigma_{\balpha} \bx)}{\balpha\cdot\bx} \balpha\cdot\bxi\label{DunklDefinition},
\end{equation}
that is, it consists of a combination of a directional derivative and a sum of difference terms taken along the reflections given by the roots in the positive subsystem. If we choose $\bxi=\be_i$, we use the notation $T_{\be_i}=T_i$, and we write
\begin{equation}
T_i f(\bx)=\frac{\partial}{\partial x_i} f(\bx)+\sum_{\balpha \in R_+}k(\balpha)\frac{f(\bx)-f(\sigma_{\balpha} \bx)}{\balpha\cdot\bx} \alpha_i.
\end{equation}
The Dunkl Laplacian is given by \cite{dunklxu}
\begin{eqnarray}
\fl\sum_{i=1}^NT_i^2 f(\bx)&=&\Delta^{(x)} f(\bx)+2\sum_{\balpha\in R_+}k(\balpha)\frac{\balpha\cdot\bnabla^{(x)} f(\bx)}{\balpha\cdot\bx}-\sum_{\balpha\in R_+}k(\balpha)\alpha^2\frac{f(\bx)-f(\sigma_{\balpha} \bx)}{(\balpha\cdot\bx)^2},
\end{eqnarray}
and with it, the Dunkl heat equation is defined by \cite{rosler98}
\begin{equation}
\frac{\partial}{\partial t}u(t,\bx)=\frac{1}{2}\sum_{i=1}^NT_i^2 u(t,\bx),
\end{equation}
where $u(t,\bx)$ is a sufficiently well-behaved function. Dunkl processes are defined as the stochastic processes whose KBE is the Dunkl heat equation. That is, their transition probability density (TPD) $p_k(t,\bx|\bx^\prime)$ obeys the equation \cite{dunklxu,roslervoit98}
\begin{eqnarray}
\frac{\partial}{\partial t}p_k(t,\bx|\bx^\prime)=&&\frac{1}{2}\Delta^{(x^\prime)} p_k(t,\bx|\bx^\prime)+\sum_{\balpha\in R_+}k(\balpha)\frac{\balpha\cdot\bnabla^{(x^\prime)} p_k(t,\bx|\bx^\prime)}{\balpha\cdot\bx^\prime}\nonumber\\
&&-\sum_{\balpha\in R_+}k(\balpha)\frac{\alpha^2}{2}\frac{p_k(t,\bx|\bx^\prime)-p_k(t,\bx|\sigma_{\balpha}\bx^\prime)}{(\balpha\cdot\bx^\prime)^2},\label{dunklheat}
\end{eqnarray}
where we denote the initial condition using primed variables. Using the explicit form of $p_k(t,\bx|\bx^\prime)$ (see, \emph{e.g.}, \cite{gallardoyor08}) and the KBE above, it can be shown directly that the corresponding KFE is given by %\cite{andrausmt}
\begin{eqnarray}
\frac{\partial}{\partial t}p_k(t,\bx|\bx^\prime)=&&\frac{1}{2}\Delta^{(x)} p_k(t,\bx|\bx^\prime)-\sum_{\balpha\in R_+}k(\balpha)\frac{\balpha\cdot\bnabla^{(x)} p_k(t,\bx|\bx^\prime)}{\balpha\cdot\bx}\nonumber\\
&&+\sum_{\balpha\in R_+}k(\balpha)\frac{\alpha^2}{2}\frac{p_k(t,\bx|\bx^\prime)+p_k(t,\sigma_{\balpha}\bx|\bx^\prime)}{(\balpha\cdot\bx)^2}.\label{dunklforward}
\end{eqnarray}
From the two Kolmogorov equations, one can read off the terms on the rhs, from left to right, as a diffusion term, a drift term and a spontaneous jump term, respectively. The drift and jump terms are clearly dependent on the root system. In particular, the Dunkl process of type $A_{N-1}$ can be understood as the Brownian motion in one dimension of $N$ particles which repel each other through a logarithmic potential and exchange positions spontaneously \cite{andrauskatorimiyashita12}. This interpretation is reached by writing down the corresponding KBE \cite{dunklxu},
\begin{eqnarray}
\frac{\partial}{\partial t}p_k(t,\bx|\bx^\prime)&=&\frac{1}{2}\Delta^{(x^\prime)}p_k(t,\bx|\bx^\prime)+\sum_{i=1}^N\sum_{\substack{j=1:\cr j\neq i}}^N\frac{k}{x_i^\prime-x_j^\prime}\frac{\partial}{\partial x_i^\prime}p_k(t,\bx|\bx^\prime)\nonumber\\
&&-\frac{k}{2}\sum_{i=1}^N\sum_{\substack{j=1:\cr j\neq i}}^N\frac{p_k(t,\bx|\bx^\prime)-p_k(t,\bx|\sigma_{ij}\bx^\prime)}{(x_j^\prime-x_i^\prime)^2},\label{typeadunklkbe}
\end{eqnarray}
and noting that this equation is similar to the KBE of Dyson's Brownian motion \cite{katoritanemura07}.

Finally, we describe radial Dunkl processes. These are the stochastic processes that result from symmetrizing the initial condition of $p_k(t,\bx|\bx^\prime)$ under the reflections generated by the root system. In this case, the jump term on the KBE vanishes, and the corresponding Dunkl process becomes a continuous stochastic processes. As noted in \cite{demni08A}, the radial Dunkl process of type $A_{N-1}$ corresponds to Dyson's Brownian motion, with $\beta=2k$. Here, $\beta$ is Dyson's beta parameter, and it is equal to 1,2 and 4 for the Gaussian orthogonal, unitary and symplectic random matrix ensembles, respectively \cite{dyson62B,dyson62}.

\section{Similarity Transformation of the Calogero-Moser Systems}\label{similaritycm}

Dunkl operators have been used as a tool to prove the integrability of the CM systems \cite{forrester10}. More specifically, it has been shown under several root systems \cite{rosler98} that after applying a similarity transformation (using the ground state eigenfunction), the CM system Hamiltonian is expressed as a Dunkl Laplacian minus a term of the form $\bx\cdot\bnabla^{(x)}$. For $R=A_{N-1}$, given the function $W(\bx)=\sum_{i=1}^Nx_i^2/2-\sum_{i<j}\ln|x_i-x_j|$, the operator equation
\begin{equation}
-\rme^{kW}(\mathcal{H}_{\textrm{CM}}^{A_{N-1}}-E_{\text{CM}}^{A_{N-1}})\rme^{-kW}=\frac{1}{2}\sum_{i=1}^N T_i^2 -k\sum_{j=1}^N x_j\frac{\partial}{\partial x_j}\label{transformedhamiltonian}
\end{equation}
is obtained. Here, the ground-state energy is given by $E_{\text{CM}}^{A_{N-1}}=[kN+k^2N(N-1)]/2$. One can then proceed to find polynomial eigenfunctions for the rhs of this expression as stated in \cite{bakerdunklforrester} and shown in \cite{bakerforrester97}.

The equation above suggests that there must be a way to transform the CM Hamiltonian of type $A_{N-1}$ into an expression proportional to the Dunkl Laplacian of type $A_{N-1}$. However, we would like to obtain a relation like \eref{transformedhamiltonian} without the second term on its rhs. 

\section{Brownian Motion and the Quantum Harmonic Oscillator}\label{ideakt}

A similar relationship appears between one-dimensional Brownian motion and the quantum harmonic oscillator; consider the Hamiltonian for the harmonic oscillator with $m=\omega=\hbar=1$, $\HH{H}=-\frac{1}{2}\frac{\partial^2}{\partial x^2}+\frac{1}{2}x^2$. A quick calculation yields
\begin{equation}
-\e^{x^2/2}(\HH{H}-1/2)\e^{-x^2/2}=\frac{1}{2}\frac{\partial^2}{\partial x^2}-x\frac{\partial}{\partial x}.\label{brownianuhlenbeck}
\end{equation}
In other words, the harmonic oscillator can be transformed into a Laplacian minus the term $x\frac{\partial}{\partial x}$. The stochastic process generated by the rhs of \eref{brownianuhlenbeck} has the KBE 
\begin{equation}
\frac{\partial}{\partial t}u(t,x)=\frac{1}{2}\frac{\partial^2}{\partial x^2}u(t,x)-x\frac{\partial}{\partial x}u(t,x),\label{ornsteinuhlenbeckkbe}
\end{equation}
and it is called an Ornstein-Uhlenbeck process \cite{mahnke09}; it represents a Brownian motion confined by a quadratic potential. Therefore, we can understand the rhs of \eref{transformedhamiltonian} as the generator of a Dunkl process in a harmonic potential. 

There also exists a way to transform the KBE of a Brownian motion (the heat equation, which is also its KFE) into a form that resembles the Schr{\"o}dinger equation of the harmonic oscillator \cite{katoritanemura07}.  Consider a Brownian motion whose probability density is given by the function $u(t,x)$ with initial condition $u(0,x)=\delta(x-x^\prime)$, $t>0$. Then, its KFE is
\begin{equation}
\frac{\partial}{\partial t}u(t,x)=\frac{1}{2}\frac{\partial^2}{\partial x^2}u(t,x).\label{brownian}
\end{equation}
Applying the variable substitution
\begin{equation}
(\tau,\zeta)=\Bigg(\ln t, \frac{x}{\sqrt{2t}}\Bigg)\label{substitutionbrownian}
\end{equation}
and the similarity transformation $u(\tau,\zeta)=\exp[-(\tau+\zeta^2)/2]U(\tau,\zeta)$, \eref{brownian} becomes
\begin{equation}
\frac{\partial}{\partial \tau}U(\tau,\zeta)=-\frac{1}{2}\left(\mathcal{H}_{\text{H}}-\frac{1}{2}\right)U(\tau,\zeta),
\end{equation}
with $x$ replaced by $\zeta$ in $\HH{H}$. This last equation shows that, under the transformation considered, Brownian motion corresponds to the harmonic oscillator evolving in imaginary time. 

Intuitively speaking, this transformation succeeds in relating Brownian motion and the harmonic oscillator because of the variable transformation \eref{substitutionbrownian}. Because Brownian motion is a free diffusion process, it is unbounded in space, whereas the harmonic oscillator is not. Dimensional analysis on the heat equation indicates that $x$ scales as $\sqrt{t}$, so we expect this diffusion process to be bounded when such a scaling is applied on $x$. This scaling yields a process similar to the Ornstein-Uhlenbeck process described by \eref{ornsteinuhlenbeckkbe}. After that, applying a suitable similarity transformation produces the harmonic oscillator Hamiltonian.

\section{Dunkl Processes and CM Systems on an Arbitrary Root System}\label{resultr}

The ideas of the previous section are applicable to Dunkl processes as follows: consider now a straightforward generalization of the substitution \eref{substitutionbrownian} given by \eref{substitutionr}, that is, consider a new space variable defined by diffusion-scaling the original space variable. Because the scaling is isotropic, it should have the same effect of \eref{substitutionbrownian} when applied on the Dunkl process on $R$. Hence, the diffusion-scaling transformation should relate Dunkl processes and the CM systems in the same way that Brownian motion and the quantum harmonic oscillator are related, as seen in the previous section. That is, we expect to be able to transform the Dunkl process on $R$ into the CM system on the same root system using the diffusion-scaling transformation, that is, by applying the similarity transformation 
\begin{equation}
u(\tau,\bzeta)=\exp[-W(\tau,\bzeta)]U(\tau,\bzeta)\label{transformationr}
\end{equation}
with $W(\tau,\bzeta)$ given by \eref{Wr} after the substitution \eref{substitutionr}. However, we must be careful with the type of Kolmogorov equation which we choose to transform. Brownian motion has the same equation for its KBE and KFE, but \eref{brownian} represents its forward evolution, so we should transform the KFE \eref{dunklforward}. We prove that the diffusion-scaling transformation indeed transforms Dunkl processes into CM systems in the following theorem.

\begin{theorem}\label{correspondencer}
The diffusion-scaling transformation given by \eref{substitutionr}, \eref{Wr} and \eref{transformationr} transforms the Dunkl process on the root system $R$ into the Calogero-Moser system with harmonic confinement on the same root system evolving in imaginary time.
\end{theorem}
\begin{proof}
We begin by considering a Dunkl process whose distribution $u(t,\bx)$ obeys both the initial condition 
\begin{equation}
u(0,\bx)=\delta^N(\bx-\bx^\prime)=\prod_{i=1}^N\delta(x_i-x_i^\prime),\quad t>0,\label{initialcondition}
\end{equation}
and the KFE \eref{dunklforward}. We write down the derivatives in time and space in terms of the new variables as follows:
\begin{eqnarray}
\frac{\partial}{\partial t}&=&\frac{1}{2\omega t}\frac{\partial}{\partial \tau}-\frac{1}{2t}\bzeta\cdot\bnabla^{(\zeta)},\nonumber\\
\frac{\partial}{\partial x_i}&=&\frac{1}{\sqrt{2\omega t}}\frac{\partial}{\partial \zeta_i}.\label{dertransformation1}
\end{eqnarray}
Note that these derivative transformations are undefined for $t=0$. Thus, we have chosen $t>0$ to guarantee the consistency of the transformation. We insert these derivatives in \eref{dunklforward} to obtain
\begin{eqnarray}
\frac{\partial}{\partial \tau}u(\tau,\bzeta)&=&\frac{1}{2}\Delta^{(\zeta)}u(\tau,\bzeta)-\sum_{\balpha\in R_+}\frac{k(\balpha)}{\balpha\cdot\bzeta}\balpha\cdot\bnabla^{(\zeta)}u(\tau,\bzeta)\nonumber\\
&&+\sum_{\balpha\in R_+}k(\balpha)\frac{\alpha^2}{2}\frac{u(\tau,\bzeta)+u(\tau,\sigma_{\balpha}\bzeta)}{(\balpha\cdot\bzeta)^2}+\omega\bzeta\cdot\bnabla^{(\zeta)}u(\tau,\bzeta).\label{rdunklsubstituted}
\end{eqnarray}
The differential operators above are transformed by \eref{transformationr} as follows:
\begin{eqnarray}
\e^{W}\frac{\partial}{\partial \tau}\e^{-W}&=&\frac{\partial}{\partial \tau}-\omega N,\nonumber\\
\e^{W}\frac{\partial}{\partial \zeta_i}\e^{-W}&=&\frac{\partial}{\partial \zeta_i}-\omega\zeta_i+\sum_{\balpha\in R_+}\frac{k(\balpha)}{\balpha\cdot\bzeta}\alpha_i,\nonumber\\
\e^{W}\Delta^{(\zeta)}\e^{-W}&=&\Delta^{(\zeta)}+2\Bigg(\sum_{\balpha\in R_+}\frac{k(\balpha)}{\balpha\cdot\bzeta}\balpha-\omega\bzeta\Bigg)\cdot\bnabla^{(\zeta)}+\omega^2\zeta^2-(2\gamma+N)\omega\nonumber\\
&&+\sum_{\balpha\in R_+}\sum_{\bxi\in R_+}\frac{k(\balpha)k(\bxi)}{(\balpha\cdot\bzeta)(\bxi\cdot\bzeta)}\balpha\cdot\bxi-\sum_{\balpha\in R_+}\frac{k(\balpha)}{(\balpha\cdot\bzeta)^2}\alpha^2.\label{differentialtransformationsr}
\end{eqnarray}
We insert the above into \eref{rdunklsubstituted} and obtain
\begin{eqnarray}
\frac{\partial}{\partial \tau}U(\tau,\bzeta)&=&\frac{1}{2}\Delta^{(\zeta)}U(\tau,\bzeta)+\frac{\omega}{2}[2\gamma+N-\omega\zeta^2]U(\tau,\bzeta)\nonumber\\
&&+\sum_{\balpha\in R_+}\frac{\alpha^2}{2}\frac{k(\balpha)}{(\balpha\cdot\bzeta)^2}U(\tau,\sigma_{\balpha}\bzeta)\nonumber\\
&&\ -\sum_{\balpha\in R_+}\sum_{\bxi\in R_+}\frac{\balpha\cdot\bxi}{2}\frac{k(\balpha)k(\bxi)}{(\balpha\cdot\bzeta)(\bxi\cdot\bzeta)}U(\tau,\bzeta).\label{withdoublesum}
\end{eqnarray}
The double sum in the bottom term of the equation above can be simplified because all the terms where $\balpha\neq\bxi$ cancel among themselves. We prove this in the Appendix. We finally obtain
\begin{eqnarray}
\frac{\partial}{\partial \tau}U(\tau,\bzeta)&=&\frac{1}{2}\Delta^{(\zeta)}U(\tau,\bzeta)+\frac{\omega}{2}[2\gamma+N-\omega\zeta^2]U(\tau,\bzeta)\nonumber\\
&&-\sum_{\balpha\in R_+}\frac{\alpha^2}{2}k(\balpha)\frac{k(\balpha) U(\tau,\bzeta)-U(\tau,\sigma_{\balpha}\bzeta)}{(\balpha\cdot\bzeta)^2},\label{almostdone}
\end{eqnarray}
or, denoting the ground-state energy by $E_{\text{0}}^{R}=\omega(\gamma+N/2)$ and using $\HH{CM}^{R}$ with $\bzeta$ instead of $\bx$,
\begin{equation}
-\frac{\partial}{\partial \tau}U(\tau,\bzeta)=[\HH{CM}^{R}-E_{\text{0}}^{R}]U(\tau,\bzeta),\label{done}
\end{equation}
as desired.
\end{proof}

\textbf{Remark:} this proof involves only straightforward calculations, with the notable exception of the step required to obtain \eref{almostdone}. This is perhaps the most important part of the proof, and it is not trivial. In the following section we focus on the root system $A_{N-1}$ and we provide an explicit example of the calculations above. In that case, the simplification needed to obtain \eref{almostdone} is straightforward. However, the general case is more involved, and it is detailed in the Appendix.

This theorem implies that CM systems and Dunkl processes share a common mathematical structure and therefore behave similarly. In particular, their freezing behaviour should be similar regardless of the root system under consideration. However, one must be careful when taking the freezing limit because there is more than one parameter involved, leading to multiple ways to take the limit (the multiplicity function $k(\balpha)$ represents more than one parameter in general, and there is also the harmonic frequency $\omega$.) This result also implies that the wealth of information available on the spectrum and eigenfunctions of the CM systems can be used to study Dunkl processes.

Note that Theorem~\ref{correspondencer} only requires that $\omega>0$. If $\omega=0$, there is no need to use the diffusion scaling \eref{substitutionr}, and one may simply apply a similarity transformation on the Dunkl process to obtain the unconfined CM system on the same root system. 

%%%%%%%
%%%%%%%
%%%%%%%
%%%%%%%
%%%%%%%

\section{Dunkl Process and Calogero-Moser System of Type $A_{N-1}$}\label{resulta}

Having proved Theorem~\ref{correspondencer}, we may consider the particular case of the root system $A_{N-1}$ as a corollary. In this case, we map the Dunkl process of type $A_{N-1}$ to the CM system with particle exchange interaction described by $\mathcal{H}_{\text{CM}}^{A_{N-1}}$ in \eref{xcm}, so we apply the diffusion-scaling transformation using \eref{substitutionr} with $\omega=k$ and  \eref{transformationr} with
\begin{equation}
W(\tau,\bzeta)=\frac{k}{2}\zeta^2-k\ln\prod_{1\leq i<j\leq N}|\zeta_i-\zeta_j|+kN\tau\label{Wa}
\end{equation}
plus a constant term that vanishes in the transformation. 
\begin{corollary}[Particular case of Theorem~\ref{correspondencer}]\label{correspondencea}
The diffusion-scaling transformation given by \eref{substitutionr}, \eref{transformationr} and \eref{Wa} with $\omega=k$ transforms the Dunkl process of type $A_{N-1}$ into the Calogero-Moser system with particle exchange under harmonic confinement evolving in imaginary time.
\end{corollary}

\begin{proof}
It suffices to follow the same procedure of Theorem~\ref{correspondencer} applied to $A_{N-1}$ to obtain the result. As before, we consider the distribution $u(t,\bx)$ with initial condition \eref{initialcondition}. The KFE \eref{dunklforward} becomes
\begin{eqnarray}
\frac{\partial}{\partial t}u(t,\bx)&=&\frac{1}{2}\Delta^{(x)}u(t,\bx)-\sum_{i=1}^N\sum_{\substack{j=1:\cr j\neq i}}^N\frac{k}{x_i-x_j}\frac{\partial}{\partial x_i}u(t,\bx)\nonumber\\
&&+\frac{k}{2}\sum_{i=1}^N\sum_{\substack{j=1:\cr j\neq i}}^N\frac{u(t,\bx)+u(t,\sigma_{ij}\bx)}{(x_j-x_i)^2}.\label{typeadunklkfe}
\end{eqnarray}
The expressions \eref{dertransformation1}, \eref{rdunklsubstituted}, \eref{differentialtransformationsr} and \eref{withdoublesum} become
\begin{eqnarray}
\frac{\partial}{\partial t}&=&\frac{1}{2kt}\frac{\partial}{\partial \tau}-\frac{1}{2t}\bzeta\cdot\bnabla^{(\zeta)},\nonumber\\
\frac{\partial}{\partial x_i}&=&\frac{1}{\sqrt{2kt}}\frac{\partial}{\partial \zeta_i},
\end{eqnarray}
\begin{eqnarray}
\frac{\partial}{\partial \tau}u(\tau,\bzeta)&=&\frac{1}{2}\Delta^{(\zeta)}u(\tau,\bzeta)-\sum_{i=1}^N\sum_{\substack{j=1:\cr j\neq i}}^N\frac{k}{\zeta_i-\zeta_j}\frac{\partial}{\partial \zeta_i}u(\tau,\bzeta)\nonumber\\
&&+\frac{k}{2}\sum_{i=1}^N\sum_{\substack{j=1:\cr j\neq i}}^N\frac{u(\tau,\bzeta)+u(\tau,\sigma_{ij}\bzeta)}{(\zeta_j-\zeta_i)^2}+k\bzeta\cdot\bnabla^{(\zeta)}u(\tau,\bzeta),\label{typeadunklsubstituted}
\end{eqnarray}
\begin{eqnarray}
\e^{W}\frac{\partial}{\partial \tau}\e^{-W}&=&\frac{\partial}{\partial \tau}-kN,\nonumber\\
\e^{W}\frac{\partial}{\partial \zeta_i}\e^{-W}&=&\frac{\partial}{\partial \zeta_i}-k\Bigg[\zeta_i-\sum_{\substack{j=1:\cr j\neq i}}^N\frac{1}{\zeta_i-\zeta_j}\Bigg],\nonumber\\
\e^{W}\Delta^{(\zeta)}\e^{-W}&=&\Delta^{(\zeta)}-2k\sum_{i=1}^N\Bigg[\zeta_i-\sum_{\substack{j=1:\cr j\neq i}}^N\frac{1}{\zeta_i-\zeta_j}\Bigg]\frac{\partial}{\partial \zeta_i}\nonumber\\
&&+k\Bigg[k\zeta^2-kN(N-1)-N+\sum_{i=1}^N\sum_{\substack{j=1:\cr j\neq i}}^N\frac{k-1}{(\zeta_i-\zeta_j)^2}\Bigg]\label{differentialtransformationsa}
\end{eqnarray}
and
\begin{eqnarray}
\frac{\partial}{\partial \tau}U(\tau,\bzeta)&=&\frac{1}{2}\Delta^{(\zeta)}U(\tau,\bzeta)+\frac{k}{2}[kN(N-1)+N-k\zeta^2]U(\tau,\bzeta)\nonumber\\
&&+\frac{k}{2}\sum_{i=1}^N\sum_{\substack{j=1:\cr j\neq i}}^N\frac{kU(\tau,\bzeta)+U(\tau,\sigma_{ij}\bzeta)}{(\zeta_i-\zeta_j)^2}\nonumber\\
&&\ -k^2\sum_{i=1}^N\sum_{\substack{j=1:\cr j\neq i}}^N\sum_{\substack{l=1:\cr l\neq i}}^N\frac{U(\tau,\bzeta)}{(\zeta_i-\zeta_j)(\zeta_i-\zeta_l)},\label{almostcma}
\end{eqnarray}
respectively. The triple sum in the bottom term is simplified by considering the case $i\neq j\neq l$, where we have the following three term sum:
\begin{eqnarray}
&&\frac{1}{(\zeta_i-\zeta_j)(\zeta_i-\zeta_l)}+\frac{1}{(\zeta_j-\zeta_l)(\zeta_j-\zeta_i)}+\frac{1}{(\zeta_l-\zeta_i)(\zeta_l-\zeta_j)}\nonumber\\
&&=\frac{\zeta_j-\zeta_l-\zeta_i+\zeta_l+\zeta_i-\zeta_j}{(\zeta_i-\zeta_j)(\zeta_i-\zeta_l)(\zeta_j-\zeta_l)}=0.
\end{eqnarray}
Therefore, the only remaining terms are those in which $i\neq j=l$; this is a particular case of the simplification of the same term in \eref{withdoublesum}, as proved in the Appendix. Then, \eref{almostdone} and \eref{done} become
\begin{eqnarray}
\frac{\partial}{\partial \tau}U(\tau,\bzeta)&=&\frac{1}{2}\Delta^{(\zeta)}U(\tau,\bzeta)-\frac{k}{2}\sum_{i=1}^N\sum_{\substack{j=1:\cr j\neq i}}^N\frac{kU(\tau,\bzeta)-U(\tau,\sigma_{ij}\bzeta)}{(\zeta_i-\zeta_j)^2}\nonumber\\
&&-\frac{k^2}{2}\zeta^2U(\tau,\bzeta)+\frac{kN}{2}[k(N-1)+1]U(\tau,\bzeta)
\end{eqnarray}
and
\begin{equation}
-\frac{\partial}{\partial \tau}U(\tau,\bzeta)=[\HH{CM}^{A_{N-1}}-E_{\text{0}}^{A_{N-1}}]U(\tau,\bzeta),
\end{equation}
respectively, proving the statement.
\end{proof}
This corollary explains the similarity of the freezing behaviour of the radial Dunkl process of type $A_{N-1}$ and the CM system on $A_{N-1}$. On one hand, the transition probability density of the radial Dunkl processes freezes to \cite{andrauskatorimiyashita12}
\begin{equation}\label{freezing}
\lim_{k\to\infty}p^{\textup{s}}_k(t,\sqrt k \bv|\bx^\prime)k^{N/2}=p^{\textup{s}}_\infty(t,\bv)=\sum_{\rho\in S_N} \delta^N[\bv-\sqrt{2t}\rho\bz_N],
\end{equation}
where $\bz_N$ is the vector whose components are the (ordered) roots of the $N$-th Hermite polynomial, defined by \cite{arfken}
\begin{equation}
H_N(x)=(-1)^N e^{x^2}\frac{\ud^N}{\ud x^N}(e^{-x^2}).\label{hermitep}
\end{equation}
On the other hand, under the same limit, the CM system with particle exchange interaction freezes to the Polychronakos-Frahm (PF) spin chain \cite{polychronakos93}. That is, the positions of the particles freeze at the roots of the $N$-th Hermite polynomial and the only remaining dynamical part of $\HH{CM}^{A_{N-1}}$ is the PF spin chain Hamiltonian,
\begin{equation}
\HH{PF}=\sum_{1\leq i<j\leq N}\frac{\sigma_{ij}}{(z_{i,N}-z_{j,N})^2}.\label{pfhamiltonian}
\end{equation}
If we recall the diffusion-scaling substitution \eref{substitutionr}, we realize that
\begin{equation}
\bzeta=\frac{\bx}{\sqrt{2kt}}=\frac{\bv}{\sqrt{2t}},
\end{equation}
because $\bx=\sqrt{k}\bv$ in \eref{freezing}. This equation is consistent with the freezing regimes of both systems ($\bzeta\to\bz_N$ and $\bv\to\sqrt{2t}\bz_N$ when the freezing limit is taken), so the factor of $\sqrt{2t}$ is accounted for by the diffusion-scaling transformation. In other words, the similar behaviour of both systems in the freezing limit goes beyond simply multiplying or dividing a factor of $\sqrt{2t}$ to the particles' frozen positions. It is a consequence of the fact that the CM system with particle exchange interaction is the diffusion-scaling transform of the Dunkl process of type $A_{N-1}$. Therefore, we can say that because one system is the diffusion-scaling transform of the other, they behave similarly at any temperature and in particular, they freeze in the same way.

\section{Concluding Remarks}\label{conclusions}

We have established a correspondence between Dunkl processes and CM systems with harmonic confinement on a line with the use of the diffusion-scaling transformation. This strategy works because the variable substitution confines the Dunkl process equally in all directions by producing a restoring drift term (the second term on the rhs of the first line of \eref{dertransformation1}). Once this drift term is present, it only remains to perform a similar operation to the one given in \eref{transformedhamiltonian} to obtain the desired result. In our case, that operation is a similarity transformation based on a function proportional to the ground-state eigenfunction of the CM system, $\exp[-W(\tau,\bxi)]$ \cite{khastgir00}. 

By establishing this correspondence, we conclude that these two types of systems share similar characteristics. In particular, their freezing limits must be similar in the sense that, if the particles in a CM system freeze at the positions $\bz$, then the scaled final position vector of the corresponding radial Dunkl process is given by $\bv=\sqrt{2t}\bz$, where we assume that $\omega=k(\balpha)=k$ and $\bv=\bx/\sqrt{k}$ before taking the limit $k\to\infty$. The particular case of the CM system of type $B_N$ is studied in \cite{yamamototsuchiya96} where it is shown that when $\omega=k_1=k_2=k$, the system freezes at the roots of the Laguerre polynomials. Therefore, we expect that the radial Dunkl process of type $B_N$ should freeze to a scaled final position proportional to the Laguerre roots. This is a calculation we plan to tackle in the near future.

Concerning the circular case, T. Kimura has pointed out that the the freezing trick has also been applied to the CMS (or circular CM) systems. In that case, the freezing trick produces spin chains of Haldane-Shastry type \cite{haldane88,shastry88}. This prompts the possibility of constructing a mapping from Dunkl processes defined on the unit circle to CMS systems similar to the diffusion-scaling transformation discussed here. At first glance, we expect the mapping to be fairly straightforward, because both systems are defined on a space of finite size. Hence, no diffusion scaling would be needed and only a similarity transformation should be required. However, circular Dunkl processes have not yet been defined to the best of our knowledge. We suspect that the trigonometric Dunkl operators considered by Cherednik \cite{cherednik91}, Heckman \cite{heckman91} and Opdam \cite{opdam95} should be at the core of such a definition. This is a problem we leave open for future work.

As noted before, there exist different ways of taking the freezing limit. So far, we have considered the cases where all parameters are locked to the same value and then brought to infinity. However, little is known about freezing limits where not all the multiplicities tend to infinity, and in consequence, their physical implications are unknown as well. They have been studied in the context of the generalized Bessel function on certain root systems \cite{roslervoit08, roslerkoornwindervoit12}. The generalized Bessel function is relevant in the sense that it is a part of the TPD of radial Dunkl processes \cite{demni08A,demni08C,rosler08}. It has been shown that this function undergoes a transition from the root system $B_N$ with multiplicities $k_1,k_2$ to the root system $A_{N-1}$ with multiplicity $k=k_2$ when $k_1$ tends to infinity and its arguments are properly scaled. It is of interest to test whether this kind of transition occurs in the context of Dunkl processes and CM systems, and we plan to investigate this matter starting with the root system $B_N$.

Finally, and following the ideas in \cite{katoritanemura07}, the fact that the CM systems with harmonic confinement are the diffusion-scaling transform of the Dunkl processes on the same root system allows us to use the quantum mechanics of the CM systems to study Dunkl processes. In particular, the harmonic confinement produces a discrete basis of eigenfunctions which can be used to represent the TPD of Dunkl processes. Hence, we believe that the existing results on the CM systems may provide further insight into the theory of Dunkl processes and their applications in physics.

%%%%%%%%%%%%%%%%%%%%%%%%%%%%%%%%%%%%%%%%%%%%%%%%
%\begin{acknowledgments}
\ack{
%%%%%%%%%%%%%%%%%%%%%%%%%%%%%%%%%%%%%%%%%%%%%%%%
SA would like to thank T. Kimura for helpful discussions on the CM systems. SA would also like to thank P. Graczyk and the organizing committee of the conference ``Harmonic Analysis and Probability'' in Angers, France (September 2-8, 2012), where part of this work was carried out. For their comments and suggestions, SA would like to thank N. Demni, L. Gallardo, A. Hardy, M. R{\"o}sler and M. Voit.
SA is supported by the Monbukagakusho: MEXT scholarship for research students.
MK is supported by
the Grant-in-Aid for Scientific Research (C)
(Grant No.21540397) from the Japan Society for
the Promotion of Science.
}
%\end{acknowledgments}
%%%%%%%%%%%%%%%%%%%%%%%%%%%%%%%%%%%%%%%%%%%%%%%%

%

\appendix

\section*{Appendix}

\setcounter{section}{1}

We prove two lemmas necessary to complete the proof of Theorem~\ref{correspondencer}. The first lemma corresponds to Theorem~4.2.4 of \cite{dunklxu}, and it involves a polynomial called the discriminant of a (reduced) root system $R$, defined by
\begin{equation}
a_{R}(\bx)=\prod_{\bv\in R_+}\bv\cdot\bx.
\end{equation}
\begin{lemma}\label{lemmadiscriminant}
The discriminant of $R$ is the minimum-order polynomial that obeys the alternating property
\begin{equation}
a_{R}(\sigma_{\bu}\bx)=-a_{R}(\bx)
\end{equation}
for all $\bu\in R$, up to a constant factor.
\end{lemma}
\begin{proof}
By definition, a root system is invariant under reflection along any of its elements. Hence, we can divide $R_+$ into the following three sets:
\begin{eqnarray}
E_1^u&=&\{\bv\in R_+:\sigma_{\bu}\bv=\bv\},\nonumber\\
E_2^u&=&\{\bv\in R_+: {}^\exists \bv^\prime\neq\bv \text{ s.t. } \sigma_{\bu}\bv=\pm \bv^\prime\}
\end{eqnarray}
and $\{\bu\}$. That is, when a root of $R_+$ is reflected along $\bu$, it is either unchanged (orthogonal), it is reflected onto another root of $R$, which is in $R_+$ or in $R_-=-R_+$, or it is turned into its negative, \emph{i.e.}, it is $\bu$ itself. By definition, the roots in $E_2^u$ obey the property that if $\bv\in E_2^u$, there is a root $\bv^\prime$ such that $\sigma_{\bu}\bv=\pm \bv^\prime$, so applying the reflection $\sigma_{\bu}$ on this last equation yields
\begin{equation}
\sigma_{\bu}\bv^\prime=\pm \bv.
\end{equation}
That is, $\bv^\prime\in E_2^u$, so there is an even number of roots in $E_2^u$, and they form pairs such that
\begin{equation}
(\sigma_{\bu}\bv\cdot\bx)(\sigma_{\bu}\bv^\prime\cdot\bx)=(\bv^\prime\cdot\bx)(\bv\cdot\bx).
\end{equation}
Now, we can write
\begin{eqnarray}
a_{R}(\sigma_{\bu}\bx)&=&\prod_{\bv\in R_+}\sigma_{\bu}\bv\cdot\bx=\sigma_{\bu}\bu\cdot\bx\prod_{\bv\in E_1^u}\sigma_{\bu}\bv\cdot\bx\prod_{\bv\in E_2^u}\sigma_{\bu}\bv\cdot\bx\nonumber\\
&=&-\bu\cdot\bx\prod_{\bv\in E_1^u}\bv\cdot\bx\prod_{\bv\in E_2^u}\bv\cdot\bx=-\prod_{\bv\in R_+}\bv\cdot\bx=-a_{R}(\bx)
\end{eqnarray}
The reason for the first equality is that $\sigma_{\bu}\bx\cdot\by=\bx\cdot\sigma_{\bu}\by$ for any vectors $\bx$ and $\by$, with $\bu\neq\bzero$. This equation holds for all $\bu\in R$. Moreover, if we were to take any one of the factors that make up $a_r(\bx)$ away, say the factor $\bu^\prime\cdot\bx$, it would no longer be an alternating polynomial along $\bu^\prime$. This completes the proof.
\end{proof}
The following lemma is a particular case of Lemma~4.4.6 of \cite{dunklxu}. It requires Lemma~\ref{lemmadiscriminant}, and it completes Theorem~\ref{correspondencer}.

\begin{lemma}\label{lemmadoublesum}
The double sum in \eref{withdoublesum} is given by
\begin{equation}
\sum_{\balpha\in R_+}\sum_{\bxi\in R_+}\balpha\cdot\bxi\frac{k(\balpha)k(\bxi)}{(\balpha\cdot\bx)(\bxi\cdot\bx)}=\sum_{\balpha\in R_+}\alpha^2\frac{k(\balpha)^2}{(\balpha\cdot\bx)^2}.\label{eqlemmadoublesum}
\end{equation}
\end{lemma}
\begin{proof}
Let us consider two different roots of $R_+$, say, $\bu$ and $\bv$. We denote by $w_{\bu\bv}$ the rotation obtained by composing the reflections along $\bu$ and $\bv$, that is, $w_{\bu\bv}=\sigma_{\bu}\sigma_{\bv}$. Let us also define the function
\begin{equation}
f(\bx, w_{\bu\bv})=\sum_{\substack{\balpha,\bxi\in R_+:\cr \sigma_{\balpha}\sigma_{\bxi}=w_{\bu\bv}}}\frac{(\balpha\cdot\bxi) k(\balpha)k(\bxi)}{(\balpha\cdot\bx)(\bxi\cdot\bx)}.
\end{equation}
We will prove that $f(\bx, w_{\bu\bv})=0$. For this purpose, we define the function
\begin{equation}
g(\bx,w_{\bu\bv})=a_{R\cap \spn(\bu,\bv)}(\bx)f(\bx, w_{\bu\bv}).
\end{equation}
Here, $\spn(\bu,\bv)$ represents the vector space (plane) generated by the pair of vectors $\{\bu,\bv\}$. Let us calculate $f(\sigma_{\bz}\bx, w_{\bu\bv})$, with $\bz\in R\cap \spn(\bu,\bv)$, that is, $\bz$ is a root contained in the plane of the rotation $w_{\bu\bv}$.
\begin{eqnarray}
\fl f(\sigma_{\bz}\bx, w_{\bu\bv})&=&\sum_{\substack{\balpha,\bxi\in R_+:\cr \sigma_{\balpha}\sigma_{\bxi}=w_{\bu\bv}}}\frac{(\balpha\cdot\bxi) k(\balpha)k(\bxi)}{(\balpha\cdot\sigma_{\bz}\bx)(\bxi\cdot\sigma_{\bz}\bx)}=\sum_{\substack{\balpha,\bxi\in R_+:\cr \sigma_{\balpha}\sigma_{\bxi}=w_{\bu\bv}}}\frac{(\sigma_{\bz}\balpha\cdot\sigma_{\bz}\bxi) k(\sigma_{\bz}\balpha)k(\sigma_{\bz}\bxi)}{(\sigma_{\bz}\balpha\cdot\bx)(\sigma_{\bz}\bxi\cdot\bx)}\nonumber\\
&=&\sum_{\substack{\balpha^\prime,\bxi^\prime\in R_+:\cr \sigma_{\sigma_{\bz}\balpha^\prime}\sigma_{\sigma_{\bz}\bxi^\prime}=w_{\bu\bv}}}\frac{(\balpha^\prime\cdot\bxi^\prime) k(\balpha^\prime)k(\bxi^\prime)}{(\balpha^\prime\cdot\bx)(\bxi^\prime\cdot\bx)}
\end{eqnarray}
On the second line, we have used the variable substitution $\sigma_{\bz}\balpha=\balpha^\prime$, and the analogue for $\bxi^\prime$. A direct calculation reveals that $\sigma_{\sigma_{\bz}\balpha^\prime}=\sigma_{\bz}\sigma_{\balpha^\prime}\sigma_{\bz}$, so
\begin{equation}
\sigma_{\sigma_{\bz}\balpha^\prime}\sigma_{\sigma_{\bz}\bxi^\prime}=\sigma_{\bz}\sigma_{\balpha^\prime}\sigma_{\bz}\sigma_{\bz}\sigma_{\bxi^\prime}\sigma_{\bz}=\sigma_{\bz}\sigma_{\balpha^\prime}\sigma_{\bxi^\prime}\sigma_{\bz},
\end{equation}
and the condition $\sigma_{\sigma_{\bz}\balpha^\prime}\sigma_{\sigma_{\bz}\bxi^\prime}=w_{\bu\bv}$ becomes
\begin{equation}
\sigma_{\bxi^\prime}\sigma_{\balpha^\prime}=w_{\bu\bv}.
\end{equation}
Hence,
\begin{equation}
f(\sigma_{\bz}\bx, w_{\bu\bv})=\sum_{\substack{\balpha^\prime,\bxi^\prime\in R_+:\cr \sigma_{\bxi^\prime}\sigma_{\balpha^\prime}=w_{\bu\bv}}}\frac{\balpha^\prime\cdot\bxi^\prime k(\balpha^\prime)k(\bxi^\prime)}{(\balpha^\prime\cdot\bx)(\bxi^\prime\cdot\bx)}=f(\bx, w_{\bu\bv})
\end{equation}
because each term in the sum is unchanged when $\bxi^\prime$ and $\balpha^\prime$ are exchanged. Now, we know that
\begin{eqnarray}
 g(\sigma_{\bz}\bx,w_{\bu\bv})&=&a_{R\cap \spn(\bu,\bv)}(\sigma_{\bz}\bx)f(\sigma_{\bz}\bx, w_{\bu\bv})\nonumber\\
&=&-a_{R\cap \spn(\bu,\bv)}(\bx)f(\bx, w_{\bu\bv})=-g(\bx,w_{\bu\bv})
\end{eqnarray}
because of Lemma~\ref{lemmadiscriminant}, so $g(\bx,w_{\bu\bv})$ also has the alternating property. However, $a_{R\cap \spn(\bu,\bv)}(\bx)$ is the alternating polynomial on $R\cap \spn(\bu,\bv)$ of minimum degree ($|R_+\cap \spn(\bu,\bv)|$), and $g(\bx,w_{\bu\bv})$ is a polynomial of smaller degree, $|R_+\cap \spn(\bu,\bv)|-2$. This contradicts Lemma~\ref{lemmadiscriminant} unless $g(\bx,w_{\bu\bv})$ is identically equal to zero. Therefore, $f(\bx, w_{\bu\bv})$ must also be equal to zero.

Repeating the argument above for all possible pairs of different $\bu,\bv\in R_+$, one can show that the terms where $\balpha\neq\bxi$ in \eref{eqlemmadoublesum} vanish, and only the terms with $\balpha=\bxi$ remain. The remaining terms do not vanish because, in the argument above, if $\bv=\bu$ then $a_{R\cap \spn(\bu,\bv)}(\bx)$ is a linear function of $\bx$. Hence, $g(\bx,w_{\bu\bv})$ is no longer a polynomial, so we cannot use Lemma~\ref{lemmadiscriminant} to arrive to a contradiction. With this, the proof is complete.
\end{proof}

\bibliography{biblio}

\end{document}